\RequirePackage{fix-cm}
\documentclass[smallextended]{svjour3}       % onecolumn (second format)
\smartqed  % flush right qed marks, e.g. at end of proof

\usepackage[T1]{fontenc}
\usepackage{amsmath}
\usepackage{amsfonts,amssymb,epsfig,multicol,multirow}
\usepackage{rotating}
\usepackage{float}
\usepackage{array}

\begin{document}
\newcommand\myeq{\stackrel{\mathclap{\normalfont\mbox{\tiny{def}}}}{=}}
\newcommand{\FF}{\mbox{$\mathbb{F}$}}
\newcommand{\m}{\mbox{$\mathbf{m}$}}
\newcommand{\e}{\mbox{$\mathbf{e}$}}
\newcommand{\ob}{\mbox{$\overline{\omega}$}}
\newcommand{\om}{\mbox{$\omega$}}
\newcommand{\Tr}{\mbox{Tr}}
\newcommand{\C}{\mbox{$\cal C$}}
\newcommand{\Cperb}{\mbox{$\C^\bot$}}
\newcommand{\ben}{\begin{equation*}}
\newcommand{\een}{\end{equation*}}
\renewcommand{\thefootnote}{\fnsymbol{footnote}}

\title{McNie2-Gabidulin: An improvement of McNie public key encryption using Gabidulin code \thanks{The work of Jon-Lark Kim was supported by Samsung Science and Technology Foundation under
Project Number SSTF-BA1602-01.}
}

\titlerunning{McNie2-Gabidulin: An improvement of McNie}        % if too long for running head

\author{Jon-Lark Kim  \and  Young-Sik Kim  \and  Lucky Galvez  \and  Myeong Jae Kim
}


\institute{J.L. Kim, L. Galvez, M.J. Kim \at
             Department of Mathematics, Sogang University, Seoul 04107, South Korea       \\
             \email{ jlkim@sogang.ac.kr, legalvez97@gmail.com,
device89@snu.ac.kr}           %  \\
           \and
           Y.S. Kim \at
              Department of Information and Communication Engineering,
Chosun University, Gwangju 61452, South Korea \\
\email{ iamyskim@chosun.ac.kr}
}

\date{Received: date / Accepted: date}

\maketitle

\begin{abstract}
McNie is a code-based public key encryption scheme submitted as a candidate to the NIST Post-Quantum Cryptography standardization \cite{pqc}. In this paper, we present McNie2-Gabidulin, an improvement of McNie. By using Gabidulin code, we eliminate the decoding failure, which is one of the limitations of the McNie public key cryptosystem that uses LRPC codes. We prove that this new cryptosystem is IND-CPA secure. Suggested parameters are also given which provides low key sizes compared to other known code based cryptosystems with zero decryption failure probability.
\end{abstract}

\section{Introduction}
McNie~\cite{pqc} is a code-based public key encryption (PKE) scheme based on the McEliece and Niederreiter cryptosystems. It was designed to be secure against known structural attacks against code-based cryptosystems. A random generator matrix is used as part of the public key which does not give any information on the private key. This random matrix is also used to mask the private key so the result is a more random matrix, rather than a parity check matrix of an equivalent code.

However, Gaborit \cite{PQCcomment} suggested a message-recovery attack which reduced the size of the random matrix. Based on this attack and an improvement of the complexity of the ISD on rank-metric codes \cite{newGRS}, the security level of McNie decreased by almost a factor of 2. For example, the parameters submitted using 4-quasi-cyclic LRPC codes for 128-bit security (NIST Category 1) can be attacked using this improved ISD with a complexity of only around $2^{98}$. So a new set of parameters were suggested which resulted in larger public key sizes. In addition, LRPC decoding is a probabilistic decoding algorithm and hence, the original parameters suggested for McNie suffer from relatively high decryption failure probability. Therefore, it is desirable to modify McNie in order to avoid Gaborit's message-recovery attack.

 In \cite{McNieKRA}, Lau, et al proposed the use of Gabidulin codes in the McNie PKE setting in order to address the issue of decryption failure. Gabidulin codes are rank metric codes that have good structures and an efficient decoding algorithm with no decoding failure. Although the McEliece cryptosystem based on Gabidulin codes is already broken by Overbeck \cite{Overbeck}, Lau, et al \cite{McNieKRA} showed that Gabidulin codes are secure to use in the McNie setting. However, their reparation is still vulnerable to the message recovery attack by Gaborit \cite{PQCcomment}. In this paper, we propose a modification of McNie which we call McNie2, which we will show to be just a generalization of McNie. We believe that this modification results in a much stronger security against known message-recovery attacks including Gaborit's attack. Moreover, McNie2 using Gabidulin code is also shown to be secure against Overbeck's attack. Furthermore, the suggested theoretical parameters achieve the lowest known key sizes for code-based public key cryptosystems without decryption failure probability.

This paper is organized as follows. We begin by discussing some preliminary concepts in rank metric codes and in particular Gabidulin codes in Section 2. Next, we introduce the McNie public key cryptosystem as it was originally submitted to the NIST Post-Quantum Cryptography standardization and a message recovery attack that compromises the security of the original proposed parameters. In Section 4, we present McNie2, a reparation of McNie which avoids the said message recovery attack. In this variant, we use Gabidulin codes to take advantage of the zero-decoding failure probability. We also show that this proposed system achieves IND-CPA security. We follow this with our proposed parameter and compare it to the original proposed parameters. Finally, we conclude with the advantages and some limitations of our new proposed scheme.

%\section{Preliminaries}
%\begin{definition}
%Let $H = [h_{ij}]$ be an $m \times n$ matrix such that $h_{ij} \in \mathbb{F}_{q^m}$. The support of $H$ is the sub-vector space of $\mathbb{F}_{q^m}$ over $\mathbb{F}_q$ spanned by the $h_{ij}$'s.
%\end{definition}

%\begin{definition}
%A  Low Rank Parity Check (LRPC) code of rank $d$, is an  $[n,k]$ code over $\mathbb{F}_{q^m}$  that has for its parity check matrix an $(n-k)\times n$ matrix $H=(h_{ij})$ such that the sub-vector space of $\mathbb{F}_{q^m}$ generated by its coefficients $h_{ij}$ has dimension at most $d$. We call this dimension the weight of $H$. Letting $F$ be the sub-vector space of $\mathbb{F}_{q^m}$ generated by the coefficients $h_{ij}$ of $H$, we denote one of its bases by $\{F_1, F_2, \ldots, F_d\}$.
%\end{definition}

%Simply speaking, the parity check matrix $H$ of an LRPC code has support of dimension at most $d$, i.e., $Supp(H) =F=\left< F_1, F_2, \ldots, F_d \right>$.
\section{Preliminaries}
In this section, we explain the necessary background for rank metric code and Gabidulin code.

\subsection{Rank metric codes}
We begin by defining the rank metric codes or simply, rank codes. Essentially, rank  codes are linear codes equipped with the rank metric, instead of the usual Hamming metric. There are two representations of rank codes, as we will see later, which  are actually related. One of them was first introduced by Delsarte in  the $matrix$ $representation$  originally as a bilinear form \cite{Del}. The other,  the $vector$ $representation$,  was  introduced by Gabidulin in his seminal paper \cite{Gab}.

\begin{definition}
$Rank$ $codes$  in $matrix$ $representation$ are subsets of the normed space $\{\mathbb{F}_{q}^{N\times n}, Rk\}$ of $N\times n$ matrices over a finite (base) field $\mathbb{F}_q$, where the norm of a matrix $M \in \mathbb{F}_q^{N\times n}$ is defined to be its algebraic rank $Rk(M)$ over $\mathbb{F}_q$. The $rank$ $distance$ $d_R(M_1,M_2)$ between two matrices $M_1$ and $M_2$ is the rank of their difference, i.e., $d_R(M_1,M_2) =Rk(M_1-M_2)$. The $rank$ $distance$ of a matrix rank code $\mathcal{M}\subset \mathbb{F}_q^{N\times n}$ is defined as the minimal pairwise distance:
$$d(\mathcal{M})=d=min(Rk(M_i - M_j) : M_i, M_j \in \mathcal{M}, i\neq j).$$
\end{definition}

\begin{definition}
$Rank$ $codes$ in $vector$ $representation$ are defined as subsets of the normed $n$-dimensional space $\{\mathbb{F}_{q^N}^{n}, Rk\}$ of length $n$ vectors over an extension field $\mathbb{F}_{q^N}$ of $\mathbb{F}_q$, where the norm of a vector $\bold{v}\in \mathbb{F}_{q^N}^{n}$ is defined to be the $column$ rank $Rk(\bold{v}|\mathbb{F}_q)$ of this vector over $\mathbb{F}_q$, i.e., the maximal number of coordinates of $\bold{v}$ which are linearly independent over the base field $\mathbb{F}_q$. The $rank$ $distance$ between two vectors $\bold{v}_1$, $\bold{v}_2$ is the column rank of their difference $Rk(\bold{v}_1-\bold{v}_2 | \mathbb{F}_q)$. The $rank$ $distance$ of a vector rank code $\mathcal{V}\subset \mathbb{F}_{q^N}^n$ is defined as the minimal pairwise distance:
$$d(\mathcal{V})=d=min(Rk(\bold{v}_i - \bold{v}_j) : \bold{v}_i, \bold{v}_j \in \mathcal{V}, i\neq j).$$
\end{definition}

Notice that for a given basis $\beta = \{ \beta_1, \beta_2, \ldots, \beta_N \}$ of $\mathbb{F}_{q^N}$ over $\mathbb{F}_q$, every vector $\bold{v}=(v_1,v_2,\ldots,v_n)  \in \mathbb{F}_{q^N}^n$ corresponds to a matrix $\bar{\bold{v}}$ whose $i^{th}$ column consists of the coefficients when $v_i$ is written in terms of the basis $\beta$. Moreover, $Rk(\bold{v})=Rk(\bar{\bold{v}})$ and this is independent of the chosen basis. Therefore, every rank code $C$ in vector representation can be expressed as a code in matrix representation  with respect to the basis $\beta$. Throughout the rest of this paper, all rank  codes being considered are in vector representation.

Another difference between rank codes and codes in the Hamming metric is the definition of the support of a codeword.

\begin{definition}\label{def3}
Let $\bold{x} = (x_1, x_2, \ldots, x_n) \in \mathbb{F}^n_{q^N}$ be a vector of rank $r$. We denote by $E$ the $\mathbb{F}_q$-sub vector space of $\mathbb{F}_{q^N}$ generated by the entries of $\bold{x}$, i.e.,  $E= \left< x_1, x_2, \ldots, x_n \right>$. The vector space $E$
is called the $support$ of $\bold{x}$.
\end{definition}

\subsection{Gabidulin codes}
Gabidulin codes are rank metric codes introduced in 1985 \cite{Gab}. They are a well-studied class of rank metric codes used for many applications. They are constructed from a set of linearly independent elements of $\mathbb{F}_{q^m}$ used to form a Moore matrix.

Let  $[i] := q^i$, the $i$th Frobenius power. Gabidulin codes are formally defined as follows.
\begin{definition}
A matrix $G \in \mathbb{F}^{k\times n}_{q^m}$ is called a Moore matrix if there exists an element $\mathbf{g} = (g_1, \ldots, g_n) \in \mathbb{F}^n_{q^m}$ such that $G = (g^{[i-1]}_j)_{i,j}$  for $1 \leq i \leq k$ and $1 \leq j \leq n$. If $\mathbf{g} $ has rank $n$, then the $[n,k]$-Gabidulin code $Gab(\mathbf{g})$ over $\mathbb{F}_{q^m}$  of dimension $k$ with generator vector $\mathbf{g}$ is the code generated by the matrix $G$.
\end{definition}
That is, if $\mathbf{g} = \left( g_1, g_2, \ldots, g_n \right)$ where $g_1, g_2, \ldots, g_n$ are linearly independent, then $Gab(\mathbf{g})$ has the following  generator matrix
{\small $$G = \left[ \begin{array}{cccc} g_1 & g_2 & \ldots & g_n \\
g_1^{[1]} & g_2^{[1]} & \ldots & g_n^{[1]} \\
\vdots & \vdots & \vdots & \vdots \\
g_1^{[k-1]} & g_2^{[k-1]} & \ldots & g_n^{[k-1]}
\end{array} \right]. $$}

Gabidulin codes have an efficient decoding algorithm that has no probability of failure. This, together with the simplicity of the structure of the generator matrix, makes it an attractive candidate for cryptographic use. However, the structure of the generator matrix also makes it vulnerable to attacks. The McEliece variant using Gabidulin codes was completely attacked by Overbeck\cite{Overbeck}.

\section{The McNie Public Key Encryption}

McNie is one of the several code-based public key encryption schemes submitted to the NIST post-quantum standardization \cite{NIST}. It claims to be resistant against known structural attacks by employing a randomly generated matrix as part of the public key and also to mask the secret key. Any code with an efficient decoding algorithm using the parity check matrix can be used for McNie.
In \cite{pqc}, a class of rank metric codes called quasi-cyclic low rank parity check (LRPC) codes are used in McNie. These codes are often used in cryptographic applications because of their simplicity and relatively no structure. However, one drawback is that there is a non-zero probability that LRPC decoding will fail. While it is not a big issue in cryptography as long as the failure probability is negligible, it is still very much desirable to design a cryptosystem that gives no decryption failure.

The general key generation, encryption and decryption steps for the McNie public key encryption are described as follows.

\noindent\fbox{ \parbox{\textwidth}{
\noindent {\centerline{\textbf{McNie Cryptosystem} }}
\medskip

\noindent \textbf{Key Generation}

Generate a random $l \times n$ generator matrix $G^\prime$ for a code over over $\mathbb{F}_{q^m}$

Generate the parity check matrix $H$ of an $[n,k]$ linear code over $\mathbb{F}_{q^m}$ with an efficient decoding algorithm $\Phi_H$ which can correct errors of (Hamming or rank) weight up to $r$.

Construct random $n \times n$ permutation matrix $P$ and $(n-k) \times (n-k)$ invertible matrix $S$ over $\mathbb{F}_{q^m}$.

Let $F = G^\prime P^{-1}H^TS$.
\begin{itemize}
\item Secret key: $(P,H,S,\Phi_H)$
\item Public key: $(G^\prime,F)$
\end{itemize}
\bigskip

\noindent \textbf{Encryption}

The sender generates a random vector $\mathbf{e}$ of weight $r$.  A message $\mathbf{m}$ is then encrypted as $Enc(\mathbf{m})=(\mathbf{c_1},\mathbf{c_2})$ where
\begin{eqnarray*}
 \mathbf{c_1} & = & \mathbf{m}G^\prime + \mathbf{e} \\
 \mathbf{c_2} & = & \mathbf{m}F
 \end{eqnarray*}
\bigskip

\noindent \textbf{Decryption}

Let the received ciphertext be $\mathbf{y} = (\mathbf{c_1},\mathbf{c_2})$. Compute $$c_1P^{-1}H^T - c_2 S^{-1} = \mathbf{e}P^{-1}H^T.$$ Apply the decryption algorithm $\Phi_H$ to obtain $\mathbf{e}P^{-1}$ and multiply by $P$ to get the error vector $\mathbf{e}$.

Finally, $\mathbf{m}$ is recovered by solving the system $\mathbf{m}G^\prime = \mathbf{c_1}-\mathbf{e}$.
} }
%\pagebreak

\subsection{Gaborit's message recovery attack}\label{subsec:GabAttack}
There is a  message recovery attack proposed by Gaborit \cite{PQCcomment} on the McNie cryptosystem that significantly reduced the security of the original suggested parameters. Notice that if $\mathbf{m} = (m_1, m_2, \ldots, m_l)$ and $F$ is of full rank, then we obtain $n-k$ linear equations of the $m_i$'s from $\mathbf{c_2} = \mathbf{m}F$. Hence, all the coordinates $m_i$'s can be expressed in terms of some fixed $l-(n-k)$ coordinates.
We can then rewrite $\mathbf{c_1}$ as $\mathbf{c_1} = \mathbf{m^\prime} G^{\prime\prime} + \mathbf{e}$ where $G^{\prime\prime}$ is of dimension $l-(n-k)$. So an attacker can use general (Hamming \cite{BothMay} or rank \cite{newGRS}) syndrome decoding on a code of dimension $l-(n-k)$ instead of a code of dimension $l$.

%%%%%%%%%%%%%%%%%%%%%%%%%%%%%%%%%%%%%%%%%%%%%%%%%%%%%%%%%%%%%%%%%%%%%%%%%%%%%%%%%%%%%%%%%%%%%%%%%%%%%%%%%%%%%%%%%%%%%%%%%%%%%%%%%%%%%%%%%%%%%%%%%%%%%%%%%%%%%%%%%%%%%%%%%%%%%%%%%%%%%%%%%%%%%%%%%
\section{McNie2 Public Key Encryption using Gabidulin code}\label{sec:McNie2-Gab}

\noindent\fbox{ \parbox{\textwidth}{
\noindent{\centerline{\textbf{McNie2-Gabidulin Cryptosystem} }}
\medskip

\noindent \textbf{Key Generation}

Generate a random vector $\mathbf{u} \in \mathbb{F}_{q^m}^n$ and generate the $l\times n$-partial circulant matrix $G^\prime$ from $\mathbf{u}$.

Let $H$ be a parity check matrix for a $[2n-k,n]$ Gabidulin code $C=Gab(\mathbf{g})$ over $\mathbb{F}_{q^m}$ generated by $\mathbf{g}$ such that $H = \left[ \begin{array}{c|c} H_1 & H_2 \end{array} \right] $ where $H_2$ is an $(n-k) \times (n-k)$ invertible matrix. Let $\Phi_H$ be an efficient decoding algorithm for $C$ using $H$, which can correct errors of weight  up to $r=\left\lfloor \frac{n-k}{2} \right\rfloor$.

Generate random $n \times n$ permutation matrix $P$.

Compute $F = G^\prime P^{-1} H_1^T (H_2^T)^{-1}$.

\begin{itemize}
\item Public Key: $(G^\prime,F)$
\item Secret Key: $(P,H,\Phi_H)$
\end{itemize}
\bigskip

\noindent \textbf{Encryption}

Generate random vectors $\mathbf{e_1} \in \mathbb{F}_{q^m}^n$ and $\mathbf{e_2} = \mathbb{F}_{q^m}^{n-k}$ such that $\mathbf{e} = (\mathbf{e_1},\mathbf{e_2})$ has weight $r$.  Compute
\begin{eqnarray*}
\mathbf{c_1} & = & \mathbf{m}G^\prime + \mathbf{e_1} \\
\mathbf{c_2} & = & \mathbf{m}F + \mathbf{e_2}.
\end{eqnarray*}
The message $\mathbf{m} \in \mathbb{F}_{q^m}^{l}$ is encrypted as $Enc(\mathbf{m}) = (\mathbf{c_1},\mathbf{c_2}).$
\bigskip

\noindent \textbf{Decryption}

Suppose the vector $\mathbf{y} = (\mathbf{c_1},\mathbf{c_2})$ is received. Compute
\begin{eqnarray*}
 \mathbf{c_1}P^{-1} H_1^T - \mathbf{c_2} H_2^T & = & \mathbf{m_1}G^\prime P^{-1} H_1^T + \mathbf{e_1} P^{-1} H_1^T - \mathbf{m} G^\prime P^{-1} H_1^T (H_2^T)^{-1} H_2^T  \\
& & - \mathbf{e_2} H_2^T \\
& = & \mathbf{e_1} P^{-1} H_1^T  - \mathbf{e_2} H_2^T \\
& = & (\mathbf{e_1} P^{-1} , -\mathbf{e_2}) \left[ \begin{array}{c} H_1^T \\ H_2^T \end{array} \right] \\
& = &  \mathbf{e^\prime} H^T
\end{eqnarray*}

Since $\mathbf{e^\prime} = (\mathbf{e_1} P^{-1} , -\mathbf{e_2}) $ is of weight $r$, the decoding algorithm $\Phi_H$ can be applied to obtain $(\mathbf{e_1^\prime},-\mathbf{e_2})$.

Apply the permutation $P$ to $\mathbf{e_1^\prime}=\mathbf{e_1}P^{-1}$ to obtain $\mathbf{e_1}$.

Finally, solve the system $\mathbf{m} G^\prime = \mathbf{c_1} - \mathbf{e_1}$ to recover $\mathbf{m}$.
}}

To avoid the attack mentioned on the previous section, we slightly modify the encryption algorithm by introducing an error $\mathbf{e_2}$ on $\mathbf{c_2}$. As a consequence, the decryption algorithm is also slightly modified. Furthermore, $H$ which is one of the secret keys is of the form $H = \left[ \begin{array}{c|c} H_1 & H_2 \end{array} \right] $ where $H_2$ is an $(n-k) \times (n-k)$ invertible matrix. Therefore, this algorithm includes Dual-Ouroboros~\cite{DualOuroboros} where $H=[H_1~|~ I_{n-k}]$ with $I_{n-k}$ is the identity matrix of order $n-k$.
We call this modified cryptosystem {\em McNie2}. Moreover, we employ Gabidulin codes to eliminate the decoding failure, as Gabidulin codes have a known deterministic decoding algorithm. We call this version McNie2-Gabidulin. The detailed key generation, encryption and decryption steps are given.
\bigskip

%In the scheme described above, we can think of $H_2^{-1}$ to be the invertible matrix $S$ in McNie. By taking $\mathbf{e_2} = \mathbf{0}$, McNie2 becomes McNie.
\subsection{Security reduction}
Indistinguishability under chosen plaintext attack (IND-CPA) is usually defined by a security game wherein an adversary $\mathcal{A}$ chooses two plaintexts $\mathbf{m}_0$ and $\mathbf{m}_1$ and sends them to the challenger who chooses $b \in \{0,1\}$ and
encrypts $\mathbf{m}_b$ into ciphertext $\mathbf{c}$ then returns $\mathbf{c}$ to $\mathcal{A}$. The adversary wins if $\mathcal{A}$ outputs $b^\prime = b$.
The advantage of an adversary $\mathcal{A}$ is defined as $\text{Adv}_\mathcal{A}^{\text{IND-CPA}}(\lambda) = | Pr[b^\prime =b]-\frac{1}{2}|$.
A public-key encryption scheme is $(t,\epsilon)$-IND-CPA secure if for any probabilistic $t$-polynomial time adversary $\mathcal{A}$, we have $\text{Adv}_\mathcal{A}^{\text{IND-CPA}}(\lambda) < \epsilon$.

In order to prove IND-CPA security of our proposed scheme, first consider the following problems.
\medskip

\noindent{\textbf{Problem 1.}}
Given an $l \times n-k$ matrix $F$ and a full rank $l \times n$ matrix $G^\prime$, find a permutation matrix $P$ and a parity matrix $H = [H_1 | H_2]$ for a Gabidulin code such that $F = G^\prime P^{-1} H_1^T(H_2^T)^{-1}$.
\medskip

\noindent{\textbf{Problem 2.}} Rank Syndrome Decoding (RSD)~\cite{newGRS}

Let $H$ be an $(n-k) \times n$ matrix over $\mathbb{F}_{q^m}$ with $k\leq n$, $\mathbf{s} \in \mathbb{F}^{n-k}_{q^m}$ and $r$ an integer. Find $\mathbf{x} \in \mathbb{F}_{q^m}^n$ such that the rank weight of $\mathbf{x} = r$ and $H\mathbf{x}^T = \mathbf{s}$.

\medskip

The first problem is a form of a matrix factorization problem. Problem 2 on the other hand is the rank metric version of the syndrome decoding (SD) problem. The RSD problem is proven hard in \cite{RSD} by a probabilistic reduction to the SD problem, which is proven NP-hard \cite{VT}.
%\noindent{\textbf{Problem 2.}} Rank Support Learning \cite{RSL}

%Given the generator matrix $G$ of a random $[n,l]$ code over $\mathbb{F}_{q^m}$ and $F = G H^T$, where $H$ is the parity check-matrix of an LRPC code, find the matrix $H$.
%\smallskip

%\noindent{\bf Remark:} In the case of quasi-cyclic LRPC codes, the RSL problem becomes the Quasi-Cyclic Rank Syndrome Decoding problem~\cite{OuroborosR}.

\begin{theorem}
The McNie2-Gabidulin PKE is IND-CPA secure under the assumption of Problems 1 and 2.
\end{theorem}

\begin{proof} First, notice that and adversary $\mathcal{A}$ breaks the scheme if he recovers the message $\mathbf{m}$ from the public keys $G^\prime$ and $F$ and the ciphertexts $\mathbf{c_1} = \mathbf{m}G^\prime + \mathbf{e_1}$ and $c_2 = \mathbf{m} F  + \mathbf{e_2}$ or is able to obtain the secret keys. These are instances of Problems 2 and 1, respectively.

We proceed with a series of games starting from an honest run of the scheme to the case when the ciphertext and the keys are random. Let $\mathcal{A}$ be a probabilistic polynomial time adversary to our scheme and consider the following games.

\begin{itemize}

\item[{$G_0$:}] This game corresponds to an honest run of the scheme.

 If $\mathcal{W}_0$ is the event that $\mathcal{A}$ wins Game $G_0$, then $\text{Adv}_\mathcal{A}^{\text{IND-CPA}}(\lambda) = | Pr[\mathcal{W}_0]-\frac{1}{2}|$.

\item[{$G_1$:}] In this game we replace the matrices $H_1$ and $H_2$ by a random matrices $\mathcal{H}_1$ and $\mathcal{H}_2$ of the same sizes. This results to a random $F$.

 If $\mathcal{W}_1$ is the event that $\mathcal{A}$ wins Game $G_1$, then, under the assumption of Problem 1, the two games $G_0$ and $G_1$ are indistinguishable with $|Pr[\mathcal{W}_1] -Pr[\mathcal{W}_0]| < \epsilon_1$.

\item[{$G_2$:}] In this game, we modify the previous game by picking random vectors $(\mathbf{s_1},\mathbf{s_2})$ to replace $(\mathbf{c_1},\mathbf{c_2}) $.

 The adversary knows $$\left[ \begin{array}{c} \mathbf{s_1}^T \\ \mathbf{s_2}^T \end{array} \right] = \left[ \begin{array}{cc} G^\prime & F \\ I & 0 \\ 0 & I \end{array} \right] \left[\begin{array}{ccc} \mathbf{m} & \mathbf{e_1} & \mathbf{e_2} \end{array} \right]$$

\end{itemize}

 We denote by $\mathcal{W}_2$ the event that $\mathcal{A}$ wins in Game $G_2$. Since the syndrome $\mathbf{s} = \left[ \begin{array}{c} \mathbf{s_1}^T \\ \mathbf{s_2}^T \end{array} \right] $ is random,  under the RSD (Problem 2) assumption, Games $G_2$ and $G_1$ are indistinguishable with $|Pr[\mathcal{W}_2] -Pr[\mathcal{W}_1]| < \epsilon_2$.

Now, since the ciphertext challenge is random, any adversary $\mathcal{A}$ has no advantage, therefore $Pr[\mathcal{W}_2] = \frac{1}{2}$. Therefore, we have
\begin{eqnarray*}
\text{Adv}^{\text{IND-CPA}}_\mathcal{A} (\lambda) &=& |Pr[\mathcal{W}_0] - \frac{1}{2}| = \Pr[\mathcal{W}_0] - Pr[\mathcal{W}_2]| \\
& = & |Pr[\mathcal{W}_0] - Pr[\mathcal{W}_1]| + |Pr[\mathcal{W}_1] - Pr[\mathcal{W}_2]| \\
& < & \epsilon_1 + \epsilon_2 .
\end{eqnarray*}

This shows that under the assumptions of Problem 1 and 2, McNie2-Gabidulin is IND-CPA secure.
%\item[$(ii)$] Recovering the low rank parity check $H$ is enough to break the system and so given $F=GH$ and $G$, finding the matrix $H$ with low rank is an instance of the rank support learning problem.
%\end{itemize}
\end{proof}

\subsection{Practical Security}

We discuss some known attacks and how they affect the security of McNie2-Gabidulin.

\begin{enumerate}

\item \textbf{Gaborit's Attack.} This attack mentioned in Section 3.1 can be avoided in our proposed scheme because of the error $\mathbf{e_2}$ added to the ciphertext $\mathbf{c_2}$. This way, it is not possible to rewrite $\mathbf{c_1}$ in the form mentioned in Section 3.1.

\item \textbf{Key recovery attack.} Lau, et al. \cite{McNieKRA} proposed a key recovery attack to the McNie cryptosystem whenever the matrix $P$ is taken to be the identity matrix. In order to avoid this attack, we simply restrict take $P$ to be non-identity matrix so that $F = G^\prime P^{-1}H_1^T(H_2^T)^{-1}$ is a cubic multivariate system of
equations and thus have a  large solving complexity. Moreover, $H$ is not a circulant LRPC matrix and therefore this attack fails in our scheme.

\item \textbf{Overbeck's attack.} In Overbeck's attack \cite{Overbeck} on $G_{pub} = S(X | G)P$ where $S \in GL_k(\mathbb{F}_{q^m})$,  $X \in \mathbb{F}^{k\times t_1}_{q^m}$, $P \in GL_{n+t_1} (\mathbb{F}_q)$, and $G$ is a $k\times n$ Moore matrix.
Consider the code $\bar{C}$ generated by $\bar{G}^T =  (G^T_{pub}, \ldots, (G^{[n-k-1]}_{pub} )^T)$, since $G^{[i]}_{pub} = S^{[i]}(X^{[i]} | G^{[i]})P$, then $ker(\bar{C}) = n + t_1 -1$. An alternative column scrambler matrix $\bar{P}$ over $\mathbb{F}_q$ could be computed, giving $G_{pub}\bar{P}^{-1} = S(Z |G^*)$
where $G^*$ is a Moore matrix. %In our construction, $FT = STHPT (G0)T and
%(FT )[i] = (ST )[i]H[i](PT )[i]((G0)[i])T , the matrix (PT )[i]((G0)[i])T is not over
%Fq, thus we are not able to determine an alternative column scrambler matrix
% P over Fq so that FT  P􀀀1 = STH, then Overbeck's attacks fails.
It is easy to see that in our proposal, $F= G^\prime P^{-1} H_1^T (H_2^T)^{-1}$. The presence of $G^\prime$ gives an additional scrambling effect to the matrix and thus Overbeck's attack fails.

\item \textbf{Direct attack on the message.} An attacker can try to obtain the message $\mathbf{m}$ by directly attacking the ciphertext. In the rank metric, these attacks usually deploy the support of a codeword  and apply a rank metric version of the Information Set Decoding \cite{Bec}. Recent improvement of this attack is given in  \cite{newGRS}. This uses the $\mathbb{F}_{q^m}$-linearity of the code to improve complexity. The best strategy has complexity $(n-k)^3m^3q^{r\frac{(k+1)m}{n}-m}$.
\end{enumerate}

%The security of the new system is based on the following problems.
%\bigskip

%\noindent{\textbf{Problem.}} Syndrome Decoding (SD)

%Let $H$ be a $(n-k) \times n$ matrix over $\mathbb{F}_{q^m}$ with $k\leq n$, $\mathbf{s} \in \mathbb{F}^{n-k}_{q^m}$ and $r$ an integer. Find $\mathbf{x} \in \mathbb{F}_{q^m}^n$ such that  $wt(\mathbf{x}) = r$ and $H\mathbf{x}^T = \mathbf{s}$.
%\bigskip

%\noindent{\textbf{Problem}}

%Given the matrices $G^\prime$ and $F = G^\prime H_1 H_2^{-1}$, where $H_2$ is invertible, find $H_1$ and $H_2$.
%\bigskip

%This problem is a combination of the Rank Syndrome Learning (RSL) problem given in \cite{OuroborosR} and two-matrix factorization problem.

\section{Suggested parameters}
We present in Table 1 our suggested parameters for McNie2-Gabidulin for different security levels based on recent known attacks on the RSD problem \cite{newGRS}. The public key size, denoted PK, is $\frac{(n+l(n-k))m}{8}\log_2(q)$ bytes while the secret key size, denoted SK, is $\frac{(n+(2n-k))m}{8}\log_2(q)$ bytes. The ciphertext size, denoted CT, is $\frac{(n+(n-k))m}{8}\log_2(q)$ bytes. Sec in Table 1 denotes the security level in bits.

\begin{table} [h] \label{Table:param}
\caption{Suggested parameters for McNie2-Gabidulin}
\centering
{\footnotesize
\begin{tabular}{|c|c|c|c|c|c|c|c|c|c|}
\hline
Sec & $n$ & $k$ & $l$ & $q$ & $m$ & $r$ & PK & SK & CT\tabularnewline
\hline
\hline
128 & 24 & 12 & 22 & 2 & 41 & 6 & 1.476KB & 0.308KB & 0.185KB\tabularnewline
\hline
192 & 32 & 16 & 24 & 2 & 53 & 8 & 2.756KB & 0.530KB & 0.318KB\tabularnewline
\hline
256 & 36 & 18 & 29 & 2 & 59 & 9 & 4.116KB & 0.664KB & 0.399KB\tabularnewline
\hline
\end{tabular} }
\end{table}

Table 2 shows the key sizes for the suggested parameters for a version of McNie using Gabidulin codes given in \cite{McNieKRA}. By comparing  from Table 1, we can see that McNie2-Gabidulin offers smaller public and private key sizes. This is because McNie2-Gabidulin is resistant to Gaborit's attack in Section \ref{subsec:GabAttack} while McNie-Gabidulin\cite{McNieKRA} needs to increase parameters to avoid  Gaborit's attack. By using McNie2-Gabidulin, we have a gain of around 400 bytes for the 128 bit and 192 bit security levels, and around 600 bytes for the 256 bit security.

\begin{table}[h]
\caption{Suggested parameters for McNie using Gabidulin codes \cite{McNieKRA}}
\centering
\begin{tabular}{|c|c|c|c|c|c|c|c|c|c|c|}
\hline
Sec & $m$ & $n$ & $k$ & $l$ & $t_{1}$ & $t_{2}$ & $q$ & PK & SK & CT\tabularnewline
\hline
\hline
128 & 43 & 38 & 14 & 37 & 9 & 3 & 2 & 1.88KB & 0.98KB & 0.33KB\tabularnewline
\hline
129 & 44 & 40 & 14 & 38 & 10 & 3 & 2 & 1.94KB & 1.07KB & 0.36KB\tabularnewline
\hline
192 & 50 & 45 & 19 & 44 & 10 & 3 & 2 & 3.21KB & 1.36KB & 0.44KB\tabularnewline
\hline
198 & 52 & 47 & 19 & 45 & 11 & 3 & 2 & 3.40KB & 1.48KB & 0.49KB\tabularnewline
\hline
257 & 57 & 52 & 20 & 51 & 13 & 3 & 2 & 4.70KB & 1.80KB & 0.60KB\tabularnewline
\hline
257 & 59 & 54 & 22 & 51 & 13 & 3 & 2 & 4.88KB & 1.94KB & 0.63KB\tabularnewline
\hline
\end{tabular}

\end{table}

Moreover, in Table 3, for every given bit security level, we give the public key sizes, in kilobytes, for different code-based public key cryptosystems with no decryption failure. The values in column 3 are for the McNie version using Gabidulin codes given in \cite{McNieKRA} (given in Table 2). Values appearing in the last two columns are from variants of the McEliece PKE using QD-Goppa \cite{Miso} and Goppa codes \cite{McE}, respectively.

\begin{table} [h]
\caption{Public key sizes (in kilobytes) of cryptosystems with no decoding failure probability}
\centering
{\footnotesize
\begin{tabular}{|c|c|c|c|c|}
\hline
Security & McNie2- & McNie- & \multicolumn{2}{c|}{McEliece}\tabularnewline
\cline{4-5}
(bit) & Gabidulin & Gabidulin & QD-Goppa & Goppa\tabularnewline
\hline
\hline
128 & {\bf 1.476} & 1.88 & 4.096 & 192.192\tabularnewline
\hline
192 & {\bf 2.756} & 3.21 & 5.632 & -\tabularnewline
\hline
256 & {\bf 4.116}  & 4.70 & 8.192 & 958.482\tabularnewline
\hline
\end{tabular}
}
\end{table}

The above table shows our proposed cryptosystem has the lowest public key sizes. Although these key sizes are relatively higher than those of cryptosystems using LRPC codes (e.g. \cite{OuroborosR,Ouroboros}), the codes have a probabilistic decoding and thus cryptosystems based on LRPC codes have a non-zero decryption failure probability. This gives McNie2-Gabidulin an advantage over other code-based cryptosystems.

% Note that Dual-Ouroboros \cite{DualOuroboros} is another example of utilization of the McNie 2 structure, which is based on LRPC codes. Surprisingly, McNie2-LRPC turns out to be a dual version of Ouroboros-R, another candidate submitted to NIST Post-Quantum Cryptography standardization. Because Dual-Ouroboros is based on LRPC codes, it has a small key size like Ouroboros-R at the cost of sacrificing non-zero probability of decoding failure.

\section{Conclusion}
We proposed McNie2-Gabidulin, a modification of the McNie public key encryption scheme that avoids the message recovery attack by Gaborit. We also address one key issue in the design of McNie, the decoding failure probability, by using Gabidulin codes. We showed that although the McEliece cryptosystem using Gabidulin codes are already proven insecure, its use in the McNie2 cryptosystem avoids known attacks and is hence secure. We also obtain relatively low key sizes of a few kilobytes.  In fact, McNie2-Gabidulin has the lowest key size among code-based public key cryptosystems with no decryption error. We believe that this makes our proposed scheme a promising candidate for post-quantum cryptography. %However, there are not much study done on this cryptosystem so very little can be said about its practical security. Therefore, its more cryptanalysis of this scheme is expected in the future.

%%%%%%%%%%%%%%%%%%%%%%%%%%%%%%%%%%%%%%%%%%%%%%%%%%%%%%%%%%%%%%%%%%%%%%%%%%%%%%%%%%%%%%%%%%%%%%%%%%%%%%%%%%%%%%%%%%%%%%%%%%%%%%%%%%%%%%%%%%%%%%%%%%%%%%%%%%%%%%%%%%%%%%%%%%%%%%%%%%%%%%%%%%%%%%%%%%%%

\end{document}